\documentclass[suppldata]{interact}

\usepackage{epstopdf}
\usepackage[caption=false]{subfig}

\usepackage[numbers,sort&compress]{natbib}
\bibpunct[, ]{[}{]}{,}{n}{,}{,}
\renewcommand\bibfont{\fontsize{10}{12}\selectfont}
\makeatletter
\def\NAT@def@citea{\def\@citea{\NAT@separator}}
\makeatother

\theoremstyle{plain}
\newtheorem{theorem}{Theorem}[section]
\newtheorem{lemma}[theorem]{Lemma}
\newtheorem{corollary}[theorem]{Corollary}
\newtheorem{proposition}[theorem]{Proposition}

\theoremstyle{definition}
\newtheorem{definition}[theorem]{Definition}

\theoremstyle{remark}
\newtheorem{remark}{Remark}

\usepackage{amssymb}
\usepackage{verbatim}
\usepackage{blkarray}
\usepackage{multirow}
\usepackage{graphicx}
\usepackage{afterpage}
\usepackage[hidelinks]{hyperref}
\usepackage{framed}
\usepackage{amsmath}
\usepackage{systeme}

\newcommand{\ket}[1]{|#1\rangle}

\newcommand{\vp}{\varphi}

\newcommand{\rk}{\operatorname{rk}}
\newcommand{\img}{\operatorname{img}}

\newcommand{\bs}[1]{\{0,1\}^{#1}}

\newcommand{\bff}[2]{f:\bs{#1}\to\bs{#2}}
\newcommand{\f}[2]{\mathbb{F}_{#1}^{#2}}
\newcommand{\w}[1]{\mathbf{#1}}

\DeclareMathOperator{\GPK}{GPK}

\begin{document}

\title{Further Applications of the Generalised Phase Kick-Back}

\author{
\name{Joaqu\'in Ossorio--Castillo\textsuperscript{{$\ast$}}\thanks{\textsuperscript{$\ast$} Email: joaquin@mestrelab.com, ORCiD: 0000-0001-8592-2263}, Ulises Pastor--D\'iaz\textsuperscript{{$\dagger$}}\thanks{\textsuperscript{$\dagger$} CONTACT: Ulises Pastor--D\'iaz. Email: ulisespastordiaz@gmail.com, ORCiD: 0000-0002-0309-7173} and Jos\'e~M. Tornero\textsuperscript{{$\ddagger$}}\thanks{\textsuperscript{$\ddagger$} Email: tornero@us.es, ORCiD: 0000-0001-5898-1049}}
\affil{\textsuperscript{$\ast,\dagger,\ddagger$} Departamento de Álgebra, Facultad de Matemáticas, Universidad de Sevilla. Avda. Reina Mercedes s/n, 41012 Sevilla (Spain).}
}

\maketitle

\begin{abstract}
In our previous work \cite{gpk}, we defined a quantum algorithmic technique known as the Generalised Phase Kick-Back, or $\GPK$, and analysed its applications in generalising some classical quantum problems, such as the Deutsch--Jozsa problem or the Bernstein--Vazirani problem. We also proved that using this technique we can solve Simon's problem in a more efficient manner.

In this paper we continue analysing the potential of this technique, defining the concept of $\mathbf{y}$-balanced functions and solving a new problem, which further generalises the generalised Deutsch--Jozsa problem (the fully balanced image problem). This problem also underlines the relation between quantum computation and Boolean function theory, and, in particular, the Walsh and Fourier--Hadamard transforms.

We finish our discussion by solving the generalised version of Simon's problem using the $\GPK$ algorithm, while analysing the efficiency of this new solution.
\end{abstract}


\begin{amscode}
68Q12 (primary); 68Q09, 81P68 (secondary).
\end{amscode}

\begin{section}{Introduction: Phase Kick-Back and Notation}

The $\GPK$ is an algorithmic technique that we introduced and studied in \cite{gpk}, where we also explored generalisations of both the Deutsch--Jozsa problem and the Bernstein--Vazirani problem and used it to solve Simon's problem in a more efficient manner. In this paper, we will further analyse this technique, presenting a new family of problems that we can solve with its help, and analysing the solution to a general version of Simon's problem. However, we must take care of some formalities before proceeding to discuss these new results.

We will use the same notation as we did in \cite{gpk}, which is that of \cite{kaye,niel}. These two books, along with \cite{via,oyt}, can be consulted for more context on the topic of quantum computing.

We will use the bra-ket or Dirac notation, in which quantum states are written as kets, $\ket{\vp}_{n,m}$, where $n$ and $m$ represent the number of qubits in the different registers. When dealing with a one-dimensional system we will omit the subindex.

For elements of the computational basis, states will be represented by binary strings $\mathbf{x}\in\bs{n}$. These strings are marked in bold to highlight their role as vectors in the space $\mathbb{F}_2^n$. Considering this structure, $\mathbf{0}$ will refer to the zero $n$-string $\mathbf{0} = 00\cdots 0$. 

Our focus will be on two main operations. We will denote by $\oplus$ the bitwise exclusive or, which corresponds to the sum in the vector space $\mathbb{F}_2^n$.
If $\mathbf{y}, \mathbf{z}\in\{0,1\}^n$ are two strings, written
$$ \mathbf{y} = y_{n-1} \ldots y_1 y_0, \quad \quad \mathbf{z} = z_{n-1} \ldots z_1 z_0, $$
then we define the exclusive or operation for strings as the bitwise exclusive or, that is,
$$ \mathbf{y} \oplus \mathbf{z} = \left( y_{n-1} \oplus z_{n-1} \right) \ldots \left( y_1 \oplus z_1 \right) \left( y_0 \oplus z_0 \right).$$

The second operation that we will take into account is the inner product in $\{0,1\}^n$, which will be noted by 
$$ \mathbf{y} \cdot \mathbf{z} = \left( y_0 \cdot z_0  \right) \oplus \ldots \oplus \left( y_{n-1}\cdot z_{n-1} \right). $$

Note that, since the xor operation is performed bitwise, we have
$$ \mathbf{x} \cdot ( \mathbf{y} \oplus \mathbf{z} ) = (\mathbf{x} \cdot \mathbf{y}) \oplus (\mathbf{x} \cdot \mathbf{z}). $$   

Concerning quantum computing, we must also recall that for any Boolean function $\bff{n}{m}$ we can construct a quantum gate $\mathbf{U}_f$ such that
$$
\mathbf{U}_f\Big(\ket{\mathbf{x}}_n\otimes\ket{\mathbf{y}}_m\Big) = \ket{\mathbf{x}}_n\otimes\ket{\mathbf{y}\oplus f(\mathbf{x})}_m.
$$
The $\GPK$ will be applied to this sort of general functions, but the phase kick-back (the technique it is based on) deals with a Boolean function $\bff{n}{}$. The first result to bear in mind, which can be found in \cite{kaye}, is the following:

\begin{lemma}
Let $\bff{n}{}$ be a Boolean function and let $\mathbf{U}_f$ be the quantum gate that computes it. Then, in the $n+1$ qubit system, vectors of the form $\ket{\mathbf{x}}_n\otimes\ket{-}$ are eigenvectors with eigenvalue $(-1)^{f(\mathbf{x})}$ for every $\mathbf{x}\in\bs{n}.$
\end{lemma}

Some classical quantum algorithms take advantage of this phenomenon---the Deutsch--Jozsa algorithm \cite{dyj}, the Bernstein--Vazirani algorithm \cite{byv1,byv2} or Grover's algorithm \cite{grover}, for example---by considering a superposition of basic states of the form
$$
\ket{\vp}_{n,1} = \frac{1}{\sqrt{N}}\sum_{\mathbf{x}\in\bs{n}} \ket{\mathbf{x}}_n\otimes\ket{-},
$$
where $N=2^n$, and applying $\mathbf{U}_f$ to it, getting a superposition of these states with the value of $f(\mathbf{x})$ encoded in the amplitude of $\ket{\mathbf{x}}$:
$$
\mathbf{U}_f\ket{\vp}_{n,1} = \mathbf{U}_f\left(\frac{1}{\sqrt{N}}\sum_{\mathbf{x}\in\bs{n}} \ket{\mathbf{x}}_n\otimes\ket{-}\right) = \frac{1}{\sqrt{N}}\sum_{\mathbf{x}\in\bs{n}} (-1)^{f(\mathbf{x})}\ket{\mathbf{x}}_n\otimes\ket{-}.
$$

We should now turn our attention to the generalised version of this technique that we defined in \cite{gpk}. To do so, we will first recall that we will now deal with Boolean functions $\bff{n}{m}$.

Let $\mathbf{U}_f$ be the quantum gate associated with $f$, $\mathbf{H}_m$ the Hadamard gate in an $m$-qubit system and $\ket{\gamma_{\mathbf{y}}}_m = \mathbf{H}_m\ket{\mathbf{y}}_m$ for a given $\mathbf{y}\in\bs{m}$. Then, we have the following analogous result, which we proved in \cite{gpk}:

\begin{lemma}
Let $\ket{\gamma_{\mathbf{y}}}_m = \mathbf{H}_m\ket{\mathbf{y}}_m$ with $\mathbf{y} \in\{0,1\}^m$. Then, for each $\mathbf{x} \in \{0,1\}^n$, the vector $\ket{\mathbf{x}}_n\otimes\ket{\gamma_{\mathbf{y}}}_m$ is an eigenvector of $\mathbf{U}_f$ with eigenvalue $(-1)^{\mathbf{y}\cdot f(\mathbf{x})}.$ 
\end{lemma}

The idea is similar to that of the one-dimensional phase kick-back, but now we have a string $\mathbf{y}\in\bs{m}$ that interacts with $f(\mathbf{x})$. Such a $\mathbf{y}$ will be called a \textit{marker}.
 
 We will finish the introduction by presenting the Generalised Phase Kick-Back algorithm or $\GPK$, which will be widely used in what remains of the text.

\begin{definition}{(Generalised Phase Kick-Back algorithm.)}
We will refer to the following algorithm as $\GPK$ algorithm for $\mathbf{y}$ or $\GPK(\mathbf{y})$. The input of the algorithm will be a given Boolean function $\bff{n}{m}$.

\vspace{5mm}

$\mathbb{STEP}$ $1$

$\ket{\varphi_0}_{n,m} = \ket{\mathbf{0}}_n\otimes\ket{\mathbf{y}}_m.$

\vspace{5mm}

$\mathbb{STEP}$ $2$

$\ket{\varphi_1}_{n,m} = \mathbf{H}_{n+m}\ket{\varphi_1}_{n,m} = \left(\displaystyle\frac{1}{\sqrt{N}}\displaystyle\sum_{\mathbf{x}\in\bs{n}} \ket{\mathbf{x}}_n\right)\otimes\ket{\gamma_{\mathbf{y}}}_m.$

\vspace{5mm}

$\mathbb{STEP}$ $3$

$\ket{\varphi_2}_{n,m} = \mathbf{U}_f\ket{\varphi_1}_{n,m} = \left(\displaystyle\frac{1}{\sqrt{N}}\displaystyle\sum_{\mathbf{x}\in\bs{n}} (-1)^{\mathbf{y}\cdot f(\mathbf{x})}\ket{\mathbf{x}}_n\right)\otimes\ket{\gamma_{\mathbf{y}}}_m.$

\vspace{5mm}

$\mathbb{STEP}$ $4$

$\ket{\varphi_3}_{n,m} = \left(\mathbf{H}^{\otimes n}\otimes \mathbf{I}^{\otimes m}\right)\ket{\varphi_2}_{n,m} = \displaystyle\frac{1}{N}\displaystyle\sum_{\mathbf{z}\in\bs{n}} \left(\displaystyle\sum_{\mathbf{x}\in\bs{n}} (-1)^{\mathbf{y}\cdot f(\mathbf{x})\oplus\mathbf{x}\cdot\mathbf{z}}\right)\ket{\mathbf{z}}_n\otimes\ket{\gamma_{\mathbf{y}}}_m.$

\vspace{5mm}

At this point, the second register can be discarded.

\vspace{5mm}

$\mathbb{STEP}$ $5$

We measure the first register and name the result $\delta$.
\end{definition}

Summing up, the amplitude associated to an element of the computational basis, $\w{z}\in\bs{n}$, before measuring will be:
\begin{equation}\label{eq:1}
    \alpha_{\w{z}} = \frac{1}{N}\displaystyle\sum_{\mathbf{x}\in\bs{n}} (-1)^{\mathbf{y}\cdot f(\mathbf{x})\oplus\mathbf{x}\cdot\mathbf{z}}.
\end{equation}
As we will remark further along the way, this is a normalised version of the Walsh transform of $f$, whose exhaustive definition and treatment can be consulted in \cite{bfc}. It is an operator on the set of Boolean functions such that, for a function $f:\bs{n}\to\bs{m}$, it is defined by
\begin{equation}\label{eq:2}
    W_f(\w{z},\w{y}) = \sum_{\w{x}\in\{0,1\}^n} (-1)^{\w{y}\cdot f(\w{x})\oplus \w{x}\cdot\w{z}}.
\end{equation}

However, for the moment we will proceed to lay the ground needed to present a new family of problems that will later be solved using this technique.
    
\end{section}

\begin{section}{$\mathbf{y}$-Balanced Functions}

We will start by putting the focus on a problem that generalises the already generalised version of the Deutsch--Jozsa problem presented in \cite{gpk}. To study this problem, we must first define the concepts of $\mathbf{y}$-balanced and $\mathbf{y}$-constant functions, which are themselves a generalisation of concepts of balanced and constant functions that were used in the Deutsch--Jozsa problem and its generalised version \cite{dyj,gpk}.

What is more, these new definitions have a deep relationship with the Fourier--Hadamard and Walsh transforms and the study of Boolean functions. We will underline these connections as we move forward, but for a general reference on these topics, \cite{bfc,mac} can be consulted.

\begin{definition}{($\mathbf{y}$-Balanced function.)}
Let $f:\{0,1\}^n\to \{0,1\}^m$ be a Boolean function and $\mathbf{y}\in \{0,1\}^m$ a string. We say that $f$ is $\mathbf{y}$-balanced if we have $f(\mathbf{x})\cdot \mathbf{y} = 0$ for half of the strings $\mathbf{x}\in\{0,1\}^n$ and $f(\mathbf{x}) \cdot \mathbf{y} = 1$ for the other half. We will also say that $\w{y}$ balances $f$.
\end{definition}

In the same way, we can define the idea of $\mathbf{y}$-constant functions.

\begin{definition}{($\mathbf{y}$-Constant function.)}
Let $f:\{0,1\}^n\to \{0,1\}^m$ be a Boolean function and $\mathbf{y}\in \{0,1\}^m$ a string. We say that $f$ is $\mathbf{y}$-constant if for every string $\mathbf{x}\in\{0,1\}^n$ the result of $f(\mathbf{x}) \cdot \mathbf{y}$ is the same. We will also say that $\w{y}$ makes $f$ constant.
\end{definition}

Let us take a look at another characterisation that can help us clarify the relation of these functions with the $\GPK$.

\begin{proposition}\label{ampys}
Let $\bff{n}{m}$ be a Boolean function. Then, $f$ is $\mathbf{y}$-balanced if and only if:
$$
\sum_{\mathbf{x}\in\bs{n}} (-1)^{f(\mathbf{x})\cdot\mathbf{y}} = 0.
$$

Similarly, $f$ is $\mathbf{y}$-constant if and only if:
$$
\left|\sum_{\mathbf{x}\in\bs{n}} (-1)^{f(\mathbf{x})\cdot\mathbf{y}} \right| = 2^n.
$$
\end{proposition}

As we can see, these definitions are not arbitrary, as the previous expressions correspond to the non-normalised amplitude of the state $\ket{\mathbf{0}}_n$ in the first $n$-qubit register of $\ket{\vp_3}_{n,m}$, that is, $\alpha_{\w{0}}$ in Equation (\ref{eq:1}). This also underlies the link between these functions and the Walsh transform, $W_f(\w{0},\w{y})$, in Equation (\ref{eq:2}). For instance, a certain $\w{y}\in\bs{m}$ makes $f$ constant if and only if it is not contained in the support of $W_f(\w{0},\cdot)$. 

Following Proposition \ref{ampys}, the next theorem becomes trivial.

\begin{theorem}{(GPK and balanced functions.)}\label{gpkbal}
Let $f:\bs{n}\to\bs{m}$ be a Boolean function and $\mathbf{y}\in\bs{m}$ for which $f$ is $\mathbf{y}$-constant. Then the output of applying the $\GPK$ algorithm to $f$ with $\mathbf{y}$ as a marker is $\mathbf{0}$.

Let $f:\bs{n}\to\bs{m}$ be a Boolean function and $\mathbf{y}\in\bs{m}$ for which $f$ is $\mathbf{y}$-balanced. Then the output of applying the $\GPK$ algorithm to $f$ with $\mathbf{y}$ as a marker is different from $\mathbf{0}$.
\end{theorem}

\begin{proof}
As noted previously, using Proposition \ref{ampys} we can link the fact that $f$ is $\mathbf{y}$-balanced or $\mathbf{y}$-constant with the amplitude of the final state of the $\GPK$ algorithm.

If the function is $\mathbf{y}$-balanced, then said amplitude is $0$, while if it is $\mathbf{y}$-constant it is $1$.
\end{proof}

Following things up, we are going to study a specific class of functions whose interest will become clear shortly.

\begin{definition}{(Fully balanced functions.)}
Let $\bff{n}{m}$ be a Boolean function. We say that it is fully balanced if for every $\mathbf{y}\in\bs{m}$, $f$ is either $\mathbf{y}$-balanced or $\mathbf{y}$-constant.
\end{definition}

These functions are actually worth our time because they have some sort of structure, as it will made evident by the next result.

\begin{theorem}\label{fullyb}
A Boolean function $\bff{n}{m}$ is fully balanced if and only if the image of $f$, $\img(f)$, is an affine space and $\#f^{-1}(\mathbf{x})$ is the same for every $\mathbf{x}\in \img(f)$.
\end{theorem}

\begin{proof}
We will proceed by induction on $m$. For $m = 1$ the result is trivial, as $\img(f)$ is either $\{0\}$, $\{1\}$ or $\{0,1\}$---which are all of them affine spaces---and in the latter case $\#f^{-1}(0)$ must be the same as $\#f^{-1}(1)$ for $1$ to balance $f$.

Suppose now that the result is true for any $m = k$. Let us prove it for $m = k+1$. Let $\w{e}_i = \underbrace{0\ldots0}_{i-1}1\underbrace{0\ldots0}_{k+1-i}$ for $i=1,\ldots,k+1$. If every $\w{e}_i$ for $i=1,\ldots,k+1$ makes $f$ constant, then the function is constant and the result is trivial. Otherwise, there must be an $i$ such that $\w{e}_i$ balances $f$. Let us suppose without loss of generality that it is $\w{e}_1$. 

We will notate as $S_0$ the set of elements $\w{x}\in\{0,1\}^n$ such that $f(\w{x})$ starts by zero, and $S_1$ the set of $\w{x}\in\{0,1\}^n$ such that $f(\w{x})$ starts by one. As $\w{e}_1$ balances $f$, we know that $\#S_0 = \#S_1 = 2^k$.

We will define now two functions, $f_0:S_0\to\{0,1\}^k$ as $f$ restricted to $S_0$, and $f_1:S_1\to\{0,1\}^k$ as $f$ restricted to $S_1$. We will prove that both $f_0$ and $f_1$ must be fully balanced. Let $\w{y}\in\{0,1\}^k$, then $\w{z}_0 = 0\w{y}\in\{0,1\}^{k+1}$---which is the concatenation of $0$ and the $\w{y}$ string---will either make constant or balance $f$. If it makes $f$ constant, then it is trivial that $\w{y}$ makes $f_0$ and $f_1$ constant. If we call $t$ the number of elements of $\w{x}\in S_0$ such that $\w{z}_0\cdot f(\w{x}) = 0$, then, for $\w{z}_0$ to balance $f$ there must be $2^k-t$ elements $\w{x'}\in S_1$ such that $\w{z}_0\cdot f(\w{x}') = 0$. If we consider now $\w{z}_1 = 1\w{y}\in\{0,1\}^{k+1}$, then, the number of of elements of $\w{x}\in S_0$ such that $\w{z}_1\cdot f(\w{x}) = 0$ is still $t$, but now the amount of elements $\w{x'}\in S_1$ such that $\w{z}_1\cdot f(\w{x}') = 0$ is also $t$, so for $\w{z}_1$ to balance $f$ or make it constant we need that $t \in\{0,2^{k-1},2^k\}$, so $\w{y}$ will either balance $f_0$ and $f_1$ or make them constant.

This also implies that if $\w{y}$ makes $f_0$ constant, then it has to also make $f_1$ constant, and conversely. We know that both $f_0$ and $f_1$ are fully balanced, and they satisfy the induction hypothesis, so $\img(f_0)$ and $\img(f_1)$ are both affine spaces such that $\#f_i^{-1}(\w{x})$ is the same for every $\w{x}\in\img(f_i)$, for $i=0,1$. Furthermore, as they both are balanced by the same strings, then $\img(f_0)$ and $\img(f_1)$ are actually two affine spaces with the same underlying vector space, which implies that $\img(f)$ is an affine space. As $f(S_0)$ and $f(S_1)$ are disjoint, then $f$ also satisfies the condition that $\#f^{-1}(\w{x})$ is the same for every $\w{x}\in\img(f)$.
\end{proof}

Before moving on, there are some observations we should make.

\begin{remark}
The proof for Theorem \ref{fullyb} that we have presented here is just one of at least two that we have developed. The other one, which uses the properties of the Walsh transform, will be presented in a different paper, together with some other results on the topic. The main reason for not presenting it here is that it requires a large introduction on Boolean function theory, which is not as relevant for the rest of the paper, but it is more relevant in the context of the Walsh transform.
\end{remark}

\begin{remark}
We should also mention that this result is not inherently new. In Chapter $13$ of \cite{mac} there is a result (Lemma $6$), which deals with this result when dealing just with sets (instead of multisets). Although this lemma does not prove the result in the general situation for sets, the proof that they present can be adapted to do so.
\end{remark}

Let us now define the main concepts that will accompany us for the reminder of this problem.

\begin{definition}{(Constant set, balancing set and balancing number.)}
Let $\bff{n}{m}$ be a Boolean function, we say that $C(f)$, its constant set, is the set of all $\mathbf{y}\in\bs{m}$ that make $f$ constant. 

We define $B(f)$, its balancing set, as the set of all $\mathbf{y}\in\bs{m}$ that balance $f$, and $b(f)$, its balancing index, as:

$$
b(f) = \frac{\#B(f)}{\#C(f)}.
$$
\end{definition}

The balancing index is an integer number, as we shall see shortly.

\begin{proposition}
Let $\bff{n}{m}$ be a Boolean function, then $C(f)$ is a vector space and $B(f)$ is a disjoint union of affine spaces whose underlying vector space is $C(f)$.
\end{proposition}

\begin{proof}
The fact that $C(f)$ is a vector space is trivial, as if $\w{y}_1,\w{y}_2\in\{0,1\}^{m}$ make $f$ constant, then $\w{y}_1\oplus\w{y}_2$ will also make it constant. On the other hand, if $\w{y}_1$ makes $f$ constant and $\w{y}_2$ balances it, then $\w{y}_1\oplus\w{y}_2$ will also balance $f$, so we have the promised structure.
\end{proof}

Incidentally, we will see after the following result that this idea will allow us to determine whether a function is fully balanced or not by knowing $b(f)$ and the rank of $\img(f)$.

The final result of this section will be that of explicitly determining the dimension and cardinality of $C(f)$ and $B(f)$ in the fully balanced case.

\begin{theorem}{(Parameters of fully balanced functions.)}\label{fbifun}
Let $\bff{n}{m}$ be a fully balanced function and let $r$ be the dimension of $\img(f)$, then the dimension of $C(f)$ is $m-r$, and thus $\#C(f) = 2^{m-r}$, $\#B(f) = (2^r-1)2^{m-r}$ and $b(f) = 2^r-1$.
\end{theorem}

\begin{proof}
The fact that the dimension of $C(f)$ is $m-r$ follows from two ideas. Firstly, that $C$ is invariant by translation---that is $C(f) = C(f\oplus\w{y})$ for every $\w{y}\in\{0,1\}^m$, where $f\oplus\w{y}:\{0,1\}^n\to\{0,1\}^m$ is the function $(f\oplus\w{y})(\w{x}) = f(\w{x})\oplus\w{y}$---and secondly, that if $\w{0}\in\img(f)$, then $C(f) = \langle\img(f)\rangle^{\perp}$, the orthogonal vector subspace to the one generated by $\img(f)$.

It is immediate then that $\#C(f) = 2^{m-r}$, which together with the fact that $f$ is fully balanced implies that $\#B(f) = 2^m-2^{m-r} = (2^{r}-1)2^{m-r}$. By definition of $b(f)$, it follows that $b(f) = 2^r-1$.
\end{proof}

In particular, we have the following.

\begin{corollary}
Let $f:\bs{n}\to\bs{m}$ be a Boolean function, and let $r$ be the dimension of $\img(f)$. Then, the function is fully balanced if and only if $b(f) = 2^r-1$.
\end{corollary}

This is the characterisation that we will use throughout the rest of the paper. Indeed, we will construct a method that, in the case of fully balanced functions, will allow us to compute $b(f)$ and $\img(f)$.

\begin{remark}
Lastly, we will use both Theorem \ref{fullyb} and Theorem \ref{fbifun} to give some examples of fully balanced functions.

Of course, generalised constant and balanced functions (which were studied in \cite{cleve,gpk}) are instances of fully balanced functions. In the constant case, the image is just an affine space of dimension zero and the balancing index is also zero, while in the balanced situation both the dimension of the image and the balancing index are one.

Another example of fully balanced functions are linear and affine functions. In the case of these functions, which are determined by an $m\times n$ matrix whose elements are in $\bs{}$ and a vector in $\bs{m}$, the dimension of the image will depend on the rank of the matrix. Indeed, if $r$ is such a rank, then the dimension of the image is $r$ and the balancing index is $2^r-1$. If we took $r = 0$ and $r = 1$, we would have constant and balanced functions respectively, but it should be noted that not every balanced function is an affine function.

Finally, we can have other kinds of fully balanced functions, $f:\bs{n}\to\bs{m}$, in which the sets $f^{-1}(\w{z})$ for $\w{z}\in\bs{m}$ do not need to have any sort of structure, as will be the case with Example 1 in Remark \ref{example1}.
\end{remark}

\end{section}

\begin{section}{The Fully Balanced Image Problem}

Firstly, we must define the problem that we will try to solve in this section.

\begin{definition}{(The Fully Balanced Image problem.)}
Let $\bff{n}{m}$ be a fully balanced Boolean function, the Fully Balanced Image problem, or FBI problem, will be that of determining the dimension of $\img(f)$.
\end{definition}

Of course, the Generalised Deutsch-Jozsa problem is an instance of the FBI problem where the dimension of $\img(f)$ can be either $0$ or $1$.

Now, let us present a method to determine the dimension of $\img(f)$---and thus all the other parameters that appear in Theorem \ref{fbifun}---of a given fully balanced function with complete certainty. It will consist of the selection of a succession of markers for which we will apply the $\GPK$ algorithm. For the remainder of the section, we will fix $r = \dim(\img(f)).$

\begin{subsection}{First marker selection algorithm}

In the first situation that we will consider, the Boolean function $\bff{n}{m}$  will be such that either $r=0$ (constant) or $r=1$ (balanced). This problem is actually the Generalised Deutsch--Jozsa problem (which was studied in \cite{gpk}), but it is important to place it in this new context.

\vspace{5mm}

\textbf{Marker Selection Algorithm 1:}

\vspace{5mm}

$\mathbb{STEP}$ $1$

$C=\varnothing.$

Take $\mathbf{y}\in\bs{m}$ such that $\mathbf{y}\notin \langle C\rangle$ and apply $\GPK(\mathbf{y})$. If the result is different from $\mathbf{0}$, then we have ended. If not, add $\mathbf{y}$ to $C$.

\vspace{5mm}

$\mathbb{STEP}$ $2$

Repeat Step $1$ until we get a result different from $\mathbf{0}$ or there are no strings left to choose. In the first situation $f$ is balanced ($r=1$) and in the former it is constant ($r=0$).

\begin{theorem}{(Correctness of Marker Selection Algorithm 1.)}
Marker Selection Algorithm 1 correctly distinguishes the cases $r=0$ and $r=1$ with certainty in at most $m$ applications of the $\GPK$.
\end{theorem}

\begin{proof}
First of all, every element in $C$ is actually also in $C(f)$ because of Theorem \ref{gpkbal}. In the case where $r=0$, we know thanks to Theorem \ref{fbifun} that the dimension of $C(f)$ is $m$ and $C(f) = \bs{m}$, while in the case where $r=1$ it is $m-1$. Thus, if we find a string that is not in $C(f)$ we immediately know that we are in the second situation. If we find that there are $m$ independent elements in $C$ then $C(f)$ has dimension $m$ and we are in the first situation.
\end{proof}

\begin{corollary}
Let $\bff{n}{m}$ be a function for which either $r=0$ or $r=r_0>0$. Then $m-r_0+1$ applications of Algorithm 1 will suffice to tell the cases apart, as in the latter possibility $\dim(C(f)) = m-r_0$.
\end{corollary}

\begin{proof}
This result is trivial, as we are storing independent elements of $C(f)$ in $C$. If $r=0$, then the dimension of $C(f)$ is $m$, while for any other $r$ the dimension of $C(f)$ is $m-r$, so if we are in the former situation, after $m-r+1$ iterations of the algorithm we should be able to distinguish the case in which we are.
\end{proof}

\end{subsection}

\begin{subsection}{Second marker selection algorithm}

The second situation that we will focus on in the direction of a general algorithm will be to consider a Boolean function $\bff{n}{m}$ for which $r$ can be either $1$ or $2$. That is, $\img(f)$ is either a one-dimensional or a two-dimensional affine space.

\vspace{5mm}

\textbf{Marker Selection Algorithm 2:}

\vspace{5mm}

$\mathbb{STEP}$ $1$

$C = \varnothing.$

$B = \varnothing.$

\vspace{5mm}

$\mathbb{STEP}$ $2$

Take $\mathbf{y}\in\bs{m}$ such that $\mathbf{y} \notin \langle C\cup B\rangle$ and apply $\GPK(\mathbf{y})$.

If the result is $\mathbf{0}$ then add $\mathbf{y}$ to $C$ and, if $\# C = m-1$, finish and return $r = 1$.

Else, run Step $3$.

\vspace{5mm}

$\mathbb{STEP}$ $3$

If $\# B = 0$, then add $\mathbf{y}$ to $B$.

Else, let $\mathbf{s}\in B$, consider $\mathbf{y}' = \mathbf{s}\oplus \mathbf{y}$ and apply $\GPK(\mathbf{y}')$. If the result is $\mathbf{0}$ then add $\mathbf{y}'$ to $C$.

Else, add both $\mathbf{y}$ and $\mathbf{y}'$ to $B$ and return $r=2$.

\vspace{5mm}

$\mathbb{STEP}$ $4$

Repeat Step $2$ until $\#C = m-1$ or $\# B = 3$.

\begin{theorem}{(Correctness of Marker Selection Algorithm 2.)}
Algorithm 2 correctly distinguishes the cases $r=1$ and $r=2$ with certainty in at most $2m-1$ applications of the $\GPK$.
\end{theorem}

\begin{proof}

Let us first analyse the constant and balancing sets for both of these situations.

If $r=1$, then $\dim(C(f)) = m-1$ and $b(f) = 1$.

If $r = 2$, then $\dim(C(f)) = m-2$ and $b(f) = 3$.

We are using $C$ to store independent elements of $C(f)$ and $B$ to store different classes of $B(f)$, which there will be $b(f)$ of. Our goal is to either obtain $\#C = m-1$ or $\#B > 1$.

We are guaranteed to get a new independent element in $C$ or $B$ after each iteration of Step $3$, since the new $\mathbf{y}$ is independent of $C\cup B$, so if for any $\mathbf{s}\in B$, if we have that $\mathbf{y}' = \mathbf{s}\oplus\mathbf{y}$ is in $C(f)$, then it is neither $\mathbf{0}$ nor dependent on $C$.

In the worst situation, we start by getting an element of $B(f)$ and then proceed to obtain independent elements of $C(f)$ via Step $3$, which require two applications of the $\GPK$. Thus, after getting $m-2$ elements of $C$, we would have computed $1+2(m-2) = 2m -3$ times the $\GPK$, and we would still need one final element in $C$ or in $B$, which would take at most two applications of the $\GPK$ via Step $3$, so the final count of $\GPK$s would be $2m-1$.
\end{proof}

As we saw in the Generalised Deutsch--Jozsa situation and will see with the final general algorithm, if we get all the elements in $C$, then we can actually compute the elements of $\img(f)$, as the solution to the system of equations given by
$$
\{\mathbf{s}\cdot \mathbf{x} = 0 \text{ for } \mathbf{s}\in C\},
$$
is actually the underlying vector space of $\img(f)$. If we wish to know $\img(f)$ we would need a final classical application of $f$ to determine the class we are in.

\end{subsection}

\begin{subsection}{General marker selection algorithm}

Finally, we will present a general algorithm for marker selection that solves our problem for every $r$ with complete certainty. As we shall see, this algorithm works with $\mathcal{O}((m-r+1)2^r-1)$ applications of the $\GPK$, so it is efficient when $r \sim \log_2(m)$.

\vspace{5mm}

\textbf{Marker Selection Algorithm 3:}

\vspace{5mm}

$\mathbb{STEP}$ $1$

$C = \varnothing.$

$B = \varnothing.$

\vspace{5mm}

$\mathbb{STEP}$ $2$

Take $\mathbf{y}\in\bs{m}$ such that $\mathbf{y}\notin \langle C\cup B\rangle$ and apply $\GPK(\mathbf{y})$.

If the result is $\mathbf{0}$, then add $\mathbf{y}$ to $C$ and repeat Step $2$.

Else, run Step $3$.

\vspace{5mm}

$\mathbb{STEP}$ $3$

If $\# B = 0$, then add $\mathbf{y}$ to $B$ and go to Step $5$.

Else, for each $\mathbf{s}\in B$, consider $\mathbf{y}_\mathbf{s} = \mathbf{s}\oplus \mathbf{y}$ and run Step $4$.

\vspace{5mm}

$\mathbb{STEP}$ $4$

Apply $\GPK(\mathbf{y}_\mathbf{s})$. If the result is $\mathbf{0}$, then add $\mathbf{y}_\mathbf{s}$ to $C$ and go to Step $5$. Else, continue with the next element of $B$ in Step $3$.

If for every $\mathbf{s}\in B$ we get a result different from $\mathbf{0}$, add $\mathbf{y}$ and every $\mathbf{y}_\mathbf{s}$ to $B$ and go to Step $5$.

\vspace{5mm}

$\mathbb{STEP}$ $5$

Repeat Step $2$ until there are no more independent elements.

\begin{remark}{(Example 1.)}\label{example1}
Let $\bff{4}{4}$ be defined by:

\vspace{5mm}

\begin{center}
\begin{tabular}{ c c c c }
 $f(0000) = 0001$ & $f(0001) = 0000$ & $f(0010) = 0000$ & $f(0011) = 1100$ \\
 $f(0100) = 0000$ & $f(0101) = 0001$ & $f(0110) = 0001$ & $f(0111) = 1101$ \\
 $f(1000) = 1100$ & $f(1001) = 1100$ & $f(1010) = 1101$ & $f(1011) = 0000$ \\
 $f(1100) = 0001$ & $f(1101) = 1100$ & $f(1110) = 1101$ & $f(1111) = 1101$.
\end{tabular}
\end{center}

\vspace{5mm}

We must note that $f$ is a function such that $S =\img(f)$ is a vector space of dimension $2$---which together with the uniform multiplicity implies in particular that $f$ is fully balanced---and thus we can use it as input to Algorithm 3. We must remember that although we are explicitly showing $f$ here, it is only for explanatory purposes, as we are supposed to work with $f$ as a black box and not know any information about it.

We begin by setting $C = \varnothing$ and $B = \varnothing$. Then, we take $\mathbf{y} = 0001$, which is independent of $C\cup B$, and apply $\GPK(0001)$ to $f$. As $f$ is $0001$-balanced, we get a result different from $\mathbf{0}$ by Theorem \ref{gpkbal}, so after Step 3 we add $0001$ to $B$.

This tells us that we are not in the constant case, i.e., the case where the rank of $f$ is $0$.

We proceed by taking a new $\mathbf{y} = 0010$ independent of $C\cup B$. As $S$ is $0010$-constant, we know that the result after applying $\GPK(0010)$ to $f$ is $\mathbf{0}$, so we add $0010$ to $C$. This eliminates the possibility of $\rk(S) = 4$.

We continue by taking a new $\mathbf{y} = 0100$ independent of $C\cup B$. Once again, $f$ is $0100$-balanced, so for each element $\mathbf{s}$ of $B = \{0001\}$ we must compute $\mathbf{y}_{\mathbf{s}} = \mathbf{y}\oplus\mathbf{s}$ and apply $\GPK(\mathbf{y}_{\mathbf{s}})$ to $f$. In this case, we only have $0100\oplus 0001 = 0101$, as $f$ is also $0101$-balanced we must add $0101$ and $0100$ to $B$. Doing so we eliminate the possibility of $\rk(S) = 1$.

We finally consider $\mathbf{y} = 1000$, which balances $f$, so again we must apply the $\GPK$ for $1001$, $1100$ and $1101$. As $1100$ makes $f$ constant, we add it to $C$ and end up getting $C = \{0010,1100\}$ and $B = \{0001,0101,0100\}$. This implies that $\rk(S) = 2$, as there are no more independent elements left.

We could follow this up by calculating the exact image of $f$. This can be done by solving the following system with four unknowns:
$$
\begin{cases}
x_2 = 0\\
x_0 \oplus x_1 = 0,
\end{cases}
$$
and computing $f(\mathbf{0})$.
\end{remark}

\begin{theorem}{(Correctness of Marker Selection Algorithm 3.)}
Algorithm 3 correctly distinguishes the cases for a general $r$ with certainty in $\mathcal{O}\left(2^r(m-r+1)-1\right)$ applications of the $\GPK$.
\end{theorem}

\begin{proof}
As before, we will store independent elements of $C(f)$ in $C$ and representatives of the different classes of $B(f)$ in $B$. Each time we store an element in $C$ we eliminate one of the possibilities for $r$ starting from $r=m$. That is, the first element in $C$ will discard $r = m$, the second will discard $r = m-1$ and so on.

Each time we store at least $2^r$ elements in $B$ we eliminate the possibility of that particular $r$, so the first element of $B$ will discard $r=0$, the second $r=1$, the fourth $r=2$ and the $2^k$-th will discard $r = k$.

We will obtain all elements in $C$ and $B$, but the most costly thing is to get first all elements in $B$, and then all the elements in $C$. Once again, after each iteration of the algorithm we are guaranteed to either get a new independent element for $C$ or $2^k$ new elements for $B$, where $k$ is $\#B$ before the iteration. This is because the new element that we introduce, $\mathbf{y}$, is independent of $C\cup B$.

It is clear now that until we get all the elements of $B$, to get the $k$-th set of elements we would need exactly $k$ applications of the $\GPK$. As there will be $2^r-1$ of these elements, we will need $2^r-1$ total applications of the $\GPK$ to exhaust $B$.

For each of the remaining $m-r$ elements of $C$, we would need at most $2^r$ appli\-ca\-tions of the $\GPK$ each, so we end up with $2^r(m-r)$. After all the calculations, we get:
$$
2^r(m-r)+2^{r}-1 = 2^r(m-r+1)-1.
$$
\end{proof}

Again, as with the previous algorithm, once we get all the elements in $C$ we can obtain the elements of $\img(f)$ by solving the system of equations and performing a final classical application of $f$.

\begin{remark}
If we wish to distinguish $r = m-1$ from $r = m$, then the algorithm actually takes $\mathcal{O}\left(2^m-1\right)$ applications of the $\GPK$.
\end{remark}

The next result will analyse the impossibility of actually solving this problem efficiently by means of a marker selection algorithm.

\begin{theorem}
There is no marker selection algorithm that can solve the problem of distinguishing the cases $r=m-1$ and $r=m$ with certainty in less than $2^m-1$ calls to the $\GPK$.
\end{theorem}

\begin{proof}
Suppose we are at the case $r = m$. Then, for any sequence of $2^m-2$ markers we would get a result different from $\mathbf{0}$ when applying the $\GPK$, but whichever is the marker $\mathbf{y}\in\bs{m}$ we left out, there is a vector space of dimension $2^{m-1}$ which would have gotten the same result. That vector space is the solution to the equation $\mathbf{y}\cdot\mathbf{x} = 0$ where $\mathbf{x}\in\bs{m}$ is the string of unknowns.
\end{proof}

\end{subsection}

\end{section}

\begin{section}{A generalised Simon's problem}

In Section 3 we analysed the capabilities of the $\GPK$ when detecting some trace of structure in the image of a Boolean function $f:\bs{n}\to\bs{m}$. After taking a look at Simon's situation---as we did in \cite{gpk}---it is natural to pose the following question: how does the $\GPK$ behave when the structure is not in the image of $f$, but rather in the sets that share the same image? 

In particular, we want to delve into Boolean functions $\bff{n}{m}$ which satisfy that there is a vector subspace $S\subset\bs{n}$ such that $f$ factorises through $S$, i.e., $\bar{f}:\bs{n}/S\to \bs{m}$ is well defined. This generalised version of Simon's problem can be found in \cite{kaye}, and is an instance of a bigger problem known as the hidden subgroup problem.

\begin{definition}{(Hidden subgroup problem.)}
Let $G$ be a group, $X$ a set, $H\leq G$ and $f:G\to X$. We say that $f$ hides $H$ if for any two given $g_1,g_2\in G$ we have that $f(g_1) = f(g_2)$ if and only if $g_1H = g_2H$. The HSP or hidden subgroup problem is that of determining $H$ using $f$ as an oracle.
\end{definition}

In particular, this generalised version of Simon's problem can be stated as follows.

\begin{definition}{(Generalised Simon's problem.)}
Let $\bff{n}{m}$ such that there is a vector subspace $S\subset\bs{n}$ for which $f(\w{x}) = f(\w{y})$ if and only if $\w{x}-\w{y}\in S$. The problem is to find such $S$.
\end{definition}

This problem is solved by Simon's algorithm in an analogous way to Simon's problem, the proof of which can be found in \cite{kaye}. The main idea is that after each iteration we are left with a superposition of all the states in $C(S)$---which is the orthogonal subspace to $S$---all of them with the same probability. Can we use the $\GPK$ to solve this problem? The answer is yes, as we will see in the following result.

\begin{proposition}
Let $\bff{n}{m}$ be a generalised Simon function with $S\subset\bs{n}$ as its hidden subgroup, and let $\w{y}\in\bs{n}\setminus \{\w{0}\}$. Then, the amplitude of $\w{z}\in\bs{n}$ in the final quantum state $\ket{\vp_4}_n$ of $\GPK(\w{y})$ is zero if and only if there is an $\w{x}\in S$ such that $\w{z}\cdot\w{x} = 1$.
\end{proposition}

\begin{proof}
If we recall that the final state of the $\GPK$ is:
$$
\ket{\varphi_4}_{n} =\frac{1}{N}\sum_{\mathbf{z}\in\{0,1\}^n} \left[\sum_{\mathbf{x}\in\{0,1\}^n} (-1)^{f(\mathbf{x})\cdot \mathbf{y} \oplus \mathbf{x}\cdot \mathbf{z}} \right]\ket{\mathbf{z}}_n,
$$
then, if we take a $\w{z}$ as described, we have that its amplitude is:
$$
\alpha_{\w{y}}(\w{z}) = \sum_{\mathbf{x}\in\{0,1\}^n} (-1)^{f(\mathbf{x})\cdot \mathbf{y} \oplus \mathbf{x}\cdot \mathbf{z}}.
$$

We know that $S$ divides $\bs{n}$ in different orbits. If we take a set of representatives of each of the orbits, $I$, we have:
$$
\alpha_{\w{y}}(\w{z}) = \sum_{\mathbf{s}\in I}\sum_{\w{x}\in S} (-1)^{f(\mathbf{s})\cdot \mathbf{y} \oplus (\w{s}\oplus\mathbf{x})\cdot \mathbf{z}} = \sum_{\mathbf{s}\in I}(-1)^{f(\mathbf{s})\cdot \mathbf{y} \oplus \w{s}\cdot \mathbf{z}}\sum_{\w{x}\in S} (-1)^{\w{x}\cdot\w{z}}.
$$

And as there is an $\w{x}_0\in S$ such that $\w{z}\cdot\w{x}_0 = 1$, then $\displaystyle\sum_{\w{x}\in S} (-1)^{\w{x}\cdot\w{z}} = 0$.
\end{proof}

As with Simon's problem, we end up with a superposition of the same states as in Simon's algorithm, but with different amplitudes and probabilities. However, we proved in \cite{gpk} the following result.

\begin{proposition}
Let $f:\{0,1\}^n\to\{0,1\}^m$ be a generalised Simon function with $S$ as the vector subspace that factors it. If we apply the $\GPK$ algorithm with random marker selection among $\{0,1\}^m$, then the probability of obtaining a given $\w{z}\in\f{2}{n}$ as a result is:
$$
p(\w{z}) = \begin{cases}
K/N & \text{if } \w{z}\cdot\w{x} = 0 \text{ for all } \w{x}\in S\\
0 & \text{otherwise,}
\end{cases}
$$
where $k$ is the dimension of $S$ and $K = 2^k$.
\end{proposition}

If we take the same approach as with Simon's problem and choose the marker $\w{y}$ at random among the nonzero strings---which was our approach in \cite{gpk}---, we will obtain an improvement over Simon's algorithm, as we are reducing the probability of getting $\w{0}$ as a result of each iteration of the $\GPK$.

\begin{proposition}
Let $f:\bs{n}\to\bs{m}$ be a generalised Simon function with $S$ as the vector subspace that factors it. If we apply the $\GPK$ algorithm with random marker selection among $\bs{m}$, then the probability of obtaining $\w{z}$ is changed to:
$$p(\w{z}) =
\begin{cases}
    (K-1)/(N-1) & \text{if } \w{z} = 0 \\
    K/(N-1) & \text{if } \w{z}\cdot\w{x} = 0 \text{ for all } \w{x}\in S \\
    0 & \text{otherwise.}
\end{cases}
$$
\end{proposition}

\begin{proof}
The proof is completely analogous to the one we performed in \cite{gpk}, with the difference that now
$$
\sum_{\w{y}\in\f{2}{n}\setminus\{\w{0}\}} p_{\w{y}}(\w{0}) = K-1,
$$
while
$$
\sum_{\w{y}\in\f{2}{n}\setminus\{\w{0}\}} p_{\w{y}}(\w{z}) = K.
$$
if $\w{z}\neq \w{0}$ and $\w{z}\cdot\w{x} = 0$ for all $\w{x}\in S$.
\end{proof}

As we can see, the bigger the dimension of $S$, the smaller the advantage obtained by using this random selection idea.

\end{section}

\begin{section}{Conclusion and Further Research}

In this paper, we have analysed the $\GPK$ algorithm when taking into account the structure of the image of $f$ by considering the Fully Balanced Image or FBI problem. This problem is important because it is easy to notice that the $\GPK$ has a very predictable behaviour if $\img(f)$ has an affine structure. However, the question remains of whether we can use our algorithm to extract information from $f$ in a more general case.

Another interesting line of research on the $\GPK$ is that of determining its behaviour when the structure is not in $\img(f)$, but in the fact that $f$ factorises through a certain $S\subset\bs{n}$. One of these situations is Simon's problem, where we have shown that we can improve Simon's algorithm and its generalisations. It does not seem unthinkable to suppose that we can further improve this algorithm by constructing a marker selection process instead of choosing them at random.

\end{section}

\section*{Acknowledgements}

This work was supported by the \emph{Ministerio de Ciencia e Innovaci\'on} under Project PID2020-114613GB-I00 (MCIN/AEI/10.13039/501100011033).

\section*{Competing interests}

The authors report that there are no competing interests to declare.

\bibliographystyle{nature}
\bibliography{GPK1}

\end{document}